\documentclass[11pt]{article}

\usepackage[margin=1in]{geometry}
\usepackage{amsmath,array,bm}
\usepackage{amsmath,amsfonts}
\usepackage{amsthm}
\usepackage{amssymb,latexsym}
\usepackage{algorithm}
\usepackage{subcaption}
\usepackage{algorithmicx}
\usepackage{marginnote}
\usepackage[]{algpseudocode}
\usepackage[dvipsnames]{xcolor}
\usepackage{graphicx}
\usepackage{caption}
\usepackage{float}
\usepackage{thmtools}
\usepackage{mathtools}
\usepackage{thm-restate}
\usepackage{bussproofs}
\usepackage[square,sort,comma,numbers]{natbib}
\usepackage{tikz}
\usepackage{xspace}
\usepackage{xcolor}
\usepackage{tipa}
\usepackage{mdframed}
\usepackage[export]{adjustbox}[2011/08/13]
\usepackage{multirow}
\usepackage{hhline}
\usepackage{todonotes}
\usepackage{tcolorbox} 
\usepackage{dirtytalk}
\usepackage{yhmath}

\usepackage[pdftex, plainpages = false, pdfpagelabels, 
                 bookmarks=false,
                 bookmarksopen = true,
                 bookmarksnumbered = true,
                 breaklinks = true,
                 linktocpage,
                 pagebackref,
                 colorlinks = true,  % was true
                 linkcolor = blue,
                 urlcolor  = blue,
                 citecolor = red,
                 anchorcolor = green,
                 hyperindex = true,
                 hyperfigures
                 ]{hyperref}

\usepackage{thmtools} 
\usepackage{thm-restate}

\newtheorem{lemma}{Lemma}

\newtheorem{observation}{Observation}

\algtext*{EndWhile}% Remove "end while" text
\algtext*{EndIf}% Remove "end if" text
\algtext*{EndFor}% Remove "end for" text

\newtheorem{fact}[lemma]{Fact}

\title{Euclidean Bottleneck Steiner Tree is Fixed-Parameter Tractable}
\author{Sayan Bandyapadhyay\thanks{Portland State University, USA. Email: \texttt{sayanb@pdx.edu}.}\ \ \ \ \ 
William Lochet\thanks{LIRMM, Université de Montpellier, CNRS, Montpellier, France. Email: \texttt{william.lochet@gmail.com}.}\ \ \ \ \ 
Daniel Lokshtanov\thanks{University of California, Santa Barbara, USA. Email: \texttt{daniello@ucsb.edu}.}
\vspace{0.2cm}
\\
Saket Saurabh\thanks{Institute of Mathematical Sciences, Chennai, India. Email: \texttt{saket@imsc.res.in}.}\ \ \ \ \ 
Jie Xue\thanks{New York University Shanghai, China. Email: \texttt{jiexue@nyu.edu}.}}

\date{} 

\theoremstyle{definition}
\newcommand{\defparopt}[4]{
\begin{tcolorbox}[colback=gray!5!white,colframe=gray!75!black]
  \vspace{-1mm}
  \begin{tabular*}{\textwidth}{@{\extracolsep{\fill}}lr} #1  & {\bf{Parameter:}} #3 \\ \end{tabular*}
  {\bf{Input:}} #2  \\
  {\bf{Output:}} #4
  \vspace{-1mm}
\end{tcolorbox}
}

\begin{document}

% \author{
% Anonymous Submission
%}

%\maketitle

%\date{}
%\author{Anonymous Authors}
%
% \author[1]{Sayan Bandyapadhyay}
% \author[1]{William Lochet}
% \author[2]{Saket Saurabh}
% \affil[1]{University of Bergen}
% \affil[2]{IMSc}
%\date{}

\maketitle
\begin{abstract}
In the Euclidean Bottleneck Steiner Tree problem, the input consists of a set of $n$ points in $\mathbb{R}^2$ called terminals and a parameter $k$, and the goal is to compute a Steiner tree that spans all the terminals and contains at most $k$ points of $\mathbb{R}^2$ as Steiner points such that the maximum edge-length of the Steiner tree is minimized, where the length of a tree edge is the Euclidean distance between its two endpoints. The problem is well-studied and is known to be NP-hard. In this paper, we give a $k^{O(k)} n^{O(1)}$-time algorithm for Euclidean Bottleneck Steiner Tree, which implies that the problem is fixed-parameter tractable (FPT). This settles an open question explicitly asked by Bae et al. [Algorithmica, 2011], who showed that the $\ell_1$ and $\ell_{\infty}$ variants of the problem are FPT. Our approach can be generalized to the problem with $\ell_p$ metric for any rational $1 \le p \le \infty$, or even other metrics on $\mathbb{R}^2$.
\end{abstract}

\section{Introduction}
Given a (finite) set $P$ of points in $\mathbb{R}^2$, a \textit{Steiner tree} on $P$ refers to a tree with node set $P \cup S$ for some finite set $S \subseteq \mathbb{R}^2$ called \textit{Steiner points}.
The length of an edge in a Steiner tree is defined as the Euclidean distance between its two endpoints.
In the \textit{Euclidean Steiner tree} (EST) problem, we are given a set $P$ of $n$ points in $\mathbb{R}^2$, and our goal is to compute a Steiner tree on $P$ such that the total length of the tree edges is minimized.
The EST problem is a classical combinatorial optimization problem whose history dates back to the 19th century (see \cite{brazil2014history} for more details).
Motivated by both theoretical interest and practical applications, EST and its variants have been studied extensively in the past decades (see \cite{du2008Steiner,kahng1994optimal} for applications to network problems, and \cite{hauptmann2013compendium} for a compendium of many Steiner tree problems).
The seminal work of Arora~\cite{arora1998polynomial} and Mitchell~\cite{DBLP:journals/siamcomp/Mitchell99} independently gave polynomial-time approximation schemes (PTAS) for the EST problem.
Besides the Euclidean setting, the Steiner tree problem can also be defined in general graphs, where we are given an edge-weighted graph $G$ and a set of \textit{terminal} vertices, and the goal is to find a minimum-weight tree in $G$ which connects all terminal vertices.
The Steiner tree problem in the graph setting was considered as one of the most fundamental NP-complete problems in the seminal paper of Karp \cite{karp1975computational}, and has been extensively studied in the context of parameterized algorithms~\cite{BjorklundHKK07,DBLP:books/sp/CyganFKLMPPS15,DowneyFbook13,DreyfusW71,FlumGrohebook,Nederlof13}; we shall briefly summarize the work on this topic at the end of this section.

In this paper, we study a min-max version of the EST problem, called \textit{Euclidean Bottleneck Steiner Tree} (EBST), which originates back to the late 80s and has several applications in VLSI design, facility location, and communication networks \cite{chiang1989powerful,du2002approximations}.

\defparopt{\textsc{Euclidean Bottleneck Steiner Tree}}{A set $P$ of $n$ points in $\mathbb{R}^2$ and an integer $k$.}{$k$}{Compute a Steiner tree on $P$ with at most $k$ Steiner points such that the maximum length of the tree edges is minimized.}

%Such a Steiner tree as above, which minimizes the maximum edge-length is called a \textit{bottleneck Steiner tree} (BST). 
%Note that, in contrast to the EST problem, here the number of Steiner points is fixed to some integer $k$.
%Indeed, it is easy to see that the maximum length of an edge can be made arbitrarily small if we are allowed to use as many Steiner points as we want, while the EST problem is already NP-hard without that restriction \cite{bern1989Steiner}.
\noindent
Note that in the EBST problem, the restriction $k$ on the number of Steiner points is necessary (while this is not the case in the EST problem), because without such a restriction, we can have a Steiner tree in which the length of every edge is arbitrarily small.
(In fact, the EST problem with a restricted number $k$ of Steiner points has also been studied.
Brazil et al.~\cite{brazil2015generalised} gave a more general $O(n^k)$ time algorithm for this variant of EST.)

The EBST problem has received much attention over years~\cite{bae2011exact,bae2010exact,chen2001approximations,chiang1989powerful,du2002approximations,ganley1996optimal,lin1999steiner,sarrafzadeh1992bottleneck,wang2002approximation}.
It is known to be NP-hard, and even hard to approximate up to a factor of $\sqrt{2}$~\cite{du2002approximations}.
On the positive side, a polynomial-time 1.866-approximation algorithm for EBST was known~\cite{wang2002approximation}.
%Li et al.~\cite{li2004approximation} gave a tight $\sqrt{2}$-approximation algorithm for a special case of EBST where two Steiner points cannot be connected by an edge.
Furthermore, different variants of EBST have been considered, including the full-tree version which requires all points in $P$ to be the leaves of the Steiner tree~\cite{abu2011euclidean,biniaz2014optimal}, the bichromatic version~\cite{abu2019bottleneck}, the variant which prohibits edges between Steiner points \cite{li2004approximation}, etc. 
A closely related problem, Euclidean Bottleneck Steiner Path, has also been studied \cite{abu2011path,hou2010optimal}.

Even though approximation algorithms were known for EBST in full generality (as well as for some special cases)~\cite{ganley1996optimal,li2004approximation,sarrafzadeh1992bottleneck,wang2002approximation}, no exact algorithm was known for over more than one decade.
Bae et al.~\cite{bae2011exact,bae2010exact} initiated the study on the exact algorithms for EBST.
They first observed that when $k \leq 2$, EBST is polynomial-time solvable~\cite{bae2010exact}.
Then in the follow-up work~\cite{bae2011exact}, they gave an $n^{O(k)}$-time algorithm for EBST, showing that the problem is polynomial-time solvable for any fixed $k$.
However, this does not settle the fixed-parameter tractability of the problem.
Bae et al. explicitly asked the following open question in~\cite{bae2011exact}.

\begin{tcolorbox}[colback=gray!5!white,colframe=gray!75!black]
{\bf Question:}  Is \textsc{Euclidean Bottleneck Steiner Tree} fixed-parameter tractable (FPT), i.e., does it admit an algorithm with running time $f(k) \cdot n^{O(1)}$?
%Does there exist an FPT algorithm for  \ebst  parameterized by $k$? 
%{\bf Question 2:} On which families of instances do \satfcc and \mcfcc admit an  FPT-approximation algorithm or even FPT-AS? 
%On which set families do \satcc and  {\sc Maximum Coverage} admit an FPT-approximation with factor  $(1-\frac{1}{e} +\epsilon)$ for %some $\epsilon>0$?
\end{tcolorbox}

%Over more than one decade, this question has remained open.
On the positive side, Bae et al.~\cite{bae2011exact} already showed that the $\ell_1$ variant (or equivalently the $\ell_\infty$ variant) of EBST is FPT.
However, as is typical in computational geometry, the problem under $\ell_1$ metric is substantially easier than that under Euclidean metric, as the $\ell_1$-disks are axis-parallel squares (rotated by an angle of $\frac{\pi}{4}$), which are usually much more tractable.
In fact, the FPT algorithm given by Bae et al.~\cite{bae2011exact} heavily relies on the geometry of squares (or more generally, rectilinear domains), and thus they failed to extend the algorithm to EBST.
For more than one decade, the above question has remained open.

%from the perspective of parameterized algorithms.
%They designed FPT algorithms for the \bst problem on $\ell_1$ and  $\ell_{\infty}$ metrics, improving  over the result of Ganley and Salowe~\cite{ganley1996optimal} . In particular, the improved running time is of the form $f(k) \cdot  n\log^2{n}$ for some function $f$ that depends only on $k$. The situation for the $\ell_2$ metric is much less understood. no exact algorithm was known for over $15$ years. That is, the  existence of an exact algorithm was open for almost $15$ years.  Bae et al.~\cite{bae2011exact}  resolved this question by designing an algorithm for \ebst running in time  $n^{O(k)}$.  In fact, their algorithm works in $\ell_{\bm{p}}$ metric for any  $1<\bm{p} < \infty$. However, this is not an FPT algorithm. They explicitly ask the following question in their paper. 

In this paper, we resolve the above question by designing the first FPT algorithm for EBST.
Specifically, we obtain the following result.
\begin{restatable}{theorem}{main}
\label{thm:ebst}
\textsc{Euclidean Bottleneck Steiner Tree} with $k$ Steiner points can be solved in $k^{O(k)}\cdot n^{O(1)}$ time.
In particular, the problem is fixed-parameter tractable.
\end{restatable}

\noindent
Theorem~\ref{thm:ebst} directly extends to the $\ell_p$ variant of EBST for any $1 \leq p \leq \infty$, or even a broader class of metrics on $\mathbb{R}^2$ (such as a linear combination of $\ell_p$ metrics).
Our algorithm combines, in a nontrivial way, previous observations for EBST, algorithmic ideas of Bae et al.~\cite{bae2011exact} for the $\ell_1$ variant, results from combinatorial geometry for Minkowski sums, and various new geometric insights to the problem.
A key ingredient of our result is a bound on the intersection complexity of certain Minkowski sums of circular domains (Lemma~\ref{lem-key}), which might be of independent interest and can possibly find further applications.

%\noindent
\subsection{Parameterized complexity for Steiner tree in graphs}
As we consider \textit{parameterized algorithms} for EBST, here we briefly overview the parameterized study on the Steiner tree problem in graph setting~\cite{DBLP:books/sp/CyganFKLMPPS15,DowneyFbook13,FlumGrohebook}. 
%The Steiner tree problem is a fundamental problem in parameterized algorithms \cite{DBLP:books/sp/CyganFKLMPPS15,DowneyFbook13,FlumGrohebook,Niedermeierbook06}.
The classic dynamic programming algorithm for Steiner tree of Dreyfus and Wagner~\cite{DreyfusW71}, with running time $3^{|T|} \cdot \log{W} \cdot n^{O(1)}$ where $|T|$ is the number of terminals, from 1971 might well be the first parameterized algorithm for {\em any} problem.
The study of parameterized algorithms for Steiner tree has led to the design of important techniques, such as fast subset convolution~\cite{BjorklundHKK07} and the use of branching walks~\cite{Nederlof13}.
Research on the parameterized complexity of Steiner tree is still on-going, with very recent significant advances for the planar version of the problem~\cite{MarxPP2017,pilipczuk_et_al:LIPIcs:2013:3947,PilipczukPSL14}.
Furthermore, algorithms for Steiner tree are frequently used as a subroutine in fixed-parameter tractable (FPT) algorithms for other problems; examples include vertex cover problems \cite{GuoNW05}, near-perfect phylogenetic tree reconstruction \cite{BlellochDHRSS06}, and connectivity augmentation problems~\cite{BasavarajuFGMRS14}.
Apart from $|T|$, another natural parameter associated with the problem is the number $k$ of Steiner vertices. It is known that Steiner tree is W[2]-hard when parametrized by $k$~\cite{DBLP:books/sp/CyganFKLMPPS15}. 

\nocite{georgakopoulos19871}

\section{Reducing to a fixed topology} \label{sec-reduction}
In the first step, we reduce the original EBST problem to a variant where the ``topology'' of the optimal Steiner tree is given to us.
This step mostly follows from the literature~\cite{bae2011exact,bae2010exact}, while we present it in a self-contained way.
%For two points $a,b \in \mathbb{R}^2$, we denote by $\mathsf{dist}(a,b)$ the Euclidean distance between $a$ and $b$.
Let $(P,k)$ be an EBST instance, where $P$ is a set of $n$ points in $\mathbb{R}^2$ and $k$ is the parameter for the number of Steiner points.
For simplicity, we call a Steiner tree on $P$ with $k$ Steiner points a {$k$-BST} if it is optimal in terms of bottleneck length, i.e., the maximum length of the tree edges is minimized.
Our goal is just to find a $k$-BST on $P$.
A \textit{minimum spanning tree} on $P$, denoted by $\text{MST}(P)$, is a spanning tree with node set $P$ that minimizes the sum of the length of its edges.
For an integer $i \in [n]$, let $\mathcal{F}_i$ be the forest obtained from $\text{MST}(P)$ by removing the $i-1$ longest edges of $\text{MST}(P)$.
Clearly, $\mathcal{F}_i$ has $i$ connected components.
For convenience, we denote by $P^2$ the set of all edges connecting two points in $P$.
Every edge of a Steiner tree on $P$ is either in $P^2$ or incident to a Steiner point.
We need the following property of $k$-BSTs.

\begin{lemma}[\cite{bae2011exact,bae2010exact}] \label{lem-BSTstructure}
There exists a $k$-BST $\mathcal{T}^*$ on $P$ satisfying the following.
\begin{itemize}
    \item For some number $K = O(k)$, $\mathcal{T}^*$ uses all edges of $\mathcal{F}_K$ but no other edges in $P^2$.
    \item Every Steiner point in $\mathcal{T}^*$ is of degree at most $5$.
\end{itemize}
%which uses all edges of $\mathcal{F}_K$ but no other edges of $\textnormal{MST}(P)$, for some $K = O(k)$.
\end{lemma}

We can afford to guess the number $K$ in the above lemma since $K = O(k)$.
Now suppose we have $K$ in hand and want to find a $k$-BST $\mathcal{T}^*$ satisfying the conditions in Lemma~\ref{lem-BSTstructure}.
Let $E(\mathcal{F}_K)$ denote the set of edges of $\mathcal{F}_K$ and $C_1,\dots,C_K$ be the connected components of $\mathcal{F}_K$.
With a little abuse of notations, we also use $C_1,\dots,C_K$ to denote the corresponding sets of points in $P$.
Thus, $C_1,\dots,C_K$ is a partition of $P$. 
Consider the tree $\mathcal{T}^*/E(\mathcal{F}_K)$ obtained from $\mathcal{T}^*$ by contracting all edges in $E(\mathcal{F}_K)$.
Clearly, $\mathcal{T}^*/E(\mathcal{F}_K)$ contains $K+k$ nodes, among which $K$ nodes correspond to the components $C_1,\dots,C_K$ and the other $k$ nodes correspond to the $k$ Steiner points.
%Also, by Lemma~\ref{lem-BSTstructure}, every edge of $\mathcal{T}^*/E(\mathcal{F}_K)$ is incident to a Steiner point and every Steiner point is of degree at most $5$ in $\mathcal{T}^*/E(\mathcal{F}_K)$.
As $K+k = O(k)$, we can afford to guess the ``topology'' of $\mathcal{T}^*/E(\mathcal{F}_K)$.
Formally, we consider all trees of $K+k$ nodes in which $K$ nodes are marked as $C_1,\dots,C_K$.
The number of such trees is $k^{O(k)}$ and one of these trees, $\mathcal{T}_\mathsf{top}$, is isomorphic to $\mathcal{T}^*/E(\mathcal{F}_K)$ where the isomorphism maps the $C_i$-vertex of $\mathcal{T}_\mathsf{top}$ to the $C_i$-vertex of $\mathcal{T}^*/E(\mathcal{F}_K)$ for every $i \in [K]$.
By trying all possibilities, we can assume that $\mathcal{T}_\mathsf{top}$ is known.
Note that by Lemma~\ref{lem-BSTstructure}, every edge of $\mathcal{T}^*/E(\mathcal{F}_K)$ is incident to a Steiner point and every Steiner point is of degree at most $5$ in $\mathcal{T}^*/E(\mathcal{F}_K)$.
So we may also assume these properties of $\mathcal{T}_\mathsf{top}$.

Let $t_1,\dots,t_k$ be the nodes of $\mathcal{T}_\mathsf{top}$ other than the ones marked by as $C_1,\dots,C_K$, which correspond to the $k$ Steiner points.
For convenience, we write $S = \{t_1,\dots,t_k\}$.
Now the problem becomes finding a map $\phi: S \rightarrow \mathbb{R}^2$ such that $\max_{(t,t') \in E(\mathcal{T}_\mathsf{top})} d_\phi(t,t')$ is minimized, where $E(\mathcal{T}_\mathsf{top})$ denotes the set of edges of $\mathcal{T}_\mathsf{top}$.
Here, $d_\phi(t,t')$ is defined as follows.
If $t,t' \in S$, then $d_\phi(t,t') = \mathsf{dist}(\phi(t),\phi(t'))$, i.e., the Euclidean distance between $\phi(t)$ and $\phi(t')$.
If $t \in S$ and $t' \notin S$, then $t'$ is marked as $C_i$ for some $i \in [K]$ and we set $d_\phi(t,t') = \min_{x \in C_i} \mathsf{dist}(\phi(t),x)$; the definition for the case $t \notin S$ and $t' \in S$ is symmetric.
The case $t,t' \notin S$ cannot happen for any $(t,t') \in E(\mathcal{T}_\mathsf{top})$, by our assumption.
We call this problem \textit{fixed-topology} EBST.
In fact, we can further reduce the optimization version of fixed-topology EBST to the decision version, which aims to check whether there exists a map $\phi: S \rightarrow \mathbb{R}^2$ such that $\max_{(t,t') \in E(\mathcal{T}_\mathsf{top})} d_\phi(t,t') \leq \lambda$ for a given value $\lambda$, or equivalently, $d_\phi(t,t') \leq \lambda$ for all $(t,t') \in E(\mathcal{T}_\mathsf{top})$.
As we demonstrate later, an FPT algorithm for the decision problem can be converted to an algorithm for the optimization problem having the same asymptotic time complexity. As demonstrated in \cite{bae2011exact}, for all $i\in [K]$, given the unique point in $C_i$ that gets connected to a Steiner point in a fixed optimal solution, the location of the optimal Steiner points can be computed in time FPT in $k$. One can easily find these unique points by brute-force, i.e., enumerating all possible $n^{O(k)}$ choices. As we aim for an FPT algorithm, we cannot use such a brute-force approach. Hence, we use a completely different approach than the approach of \cite{bae2011exact} in the $\ell_2$ metric. However, the starting point of our approach is the FPT algorithm of \cite{bae2011exact} in the $\ell_1$ metric.    
%using the standard technique of parametric search \cite{megiddo1983applying}.

%\noindent
%\noindent\fbox{%
%    \parbox{\textwidth}{%
%        \textbf{Fixed-topology EBST (Decision)} \\
%        \textbf{Input:} $K$ point-sets $C_1,\dots,C_K \subseteq \mathbb{R}^2$ with $\sum_{i=1}^K |C_i| = n$, a tree $\mathcal{T}_\mathsf{top}$ of $K+k$ nodes in which $K$ nodes are marked as $C_1,\dots,C_K$, and a number $\lambda > 0$. \\
%        \textbf{Goal:} Compute a map $\phi: S \rightarrow \mathbb{R}^2$ such that $d_\phi(t,t') \leq \lambda$ for all $(t,t') \in E(\mathcal{T}_\mathsf{top})$ where $S$ is the set of unmarked nodes of $\mathcal{T}_\mathsf{top}$, or conclude the non-existence of such a map.
%    }%
%}

\section{The main algorithm}
\label{sec:FPT}

According to the discussion in Section~\ref{sec-reduction}, it suffices to design an algorithm for (the decision version of) fixed-topology EBST.
Now we restate the setting of the problem.
We have a set $P$ of $n$ points in $\mathbb{R}^2$ which is partitioned into $C_1,\dots,C_K$, a tree $\mathcal{T}_\mathsf{top}$ of $K+k$ nodes in which $K$ nodes are marked as $C_1,\dots,C_K$ (called terminal nodes) and the other $k$ nodes are called Steiner points, and a number $\lambda > 0$.
%We use $S \subseteq \mathcal{T}_\mathsf{top}$ to denote the set of $k$ Steiner points.
The tree $\mathcal{T}_\mathsf{top}$ satisfies \textbf{(i)} every edge is incident to a Steiner point and \textbf{(ii)} every Steiner point is of degree at most 5.
With an abuse of notations, in what follows, we shall just use $C_1,\dots,C_K$ to denote the terminal nodes of $\mathcal{T}_\mathsf{top}$ (instead of saying that they are marked as $C_1,\dots,C_K$).
Our goal is to compute a map $\phi: S \rightarrow \mathbb{R}^2$ such that $d_\phi(t,t') \leq \lambda$ for all $(t,t') \in E(\mathcal{T}_\mathsf{top})$ where $S \subseteq \mathcal{T}_\mathsf{top}$ is the set of Steiner points in $\mathcal{T}_\mathsf{top}$, or conclude the non-existence of such a map.
By scaling the points in $P$, we can assume $\lambda = 1$ without loss of generality.
Furthermore, we can assume that every leaf (i.e., a node of degree 1) of $\mathcal{T}_\mathsf{top}$ is a terminal node.
Indeed, if $s \in S$ is a leaf in $\mathcal{T}_\mathsf{top}$ and $\phi_0: S \backslash \{s\} \rightarrow \mathbb{R}^2$ is a map satisfying $d_\phi(t,t') \leq \lambda$ for all $(t,t') \in E(\mathcal{T}_\mathsf{top} \backslash \{s\})$, then we can easily extend $\phi_0$ to a map $\phi: S \rightarrow \mathbb{R}^2$ satisfying the desired property by choosing the point $\phi(s)$ to make $d_\phi(s,t) = 0$ where $t$ is the only node in $\mathcal{T}_\mathsf{top}$ neighboring to $s$.
Therefore, we can simply remove $s$ from $\mathcal{T}_\mathsf{top}$ (and also from $S$).
By doing this repeatedly, we can reach the situation where every leaf of $\mathcal{T}_\mathsf{top}$ is a terminal node.

In fact, we can further reduce to the situation where the leaves of $\mathcal{T}_\mathsf{top}$ are \textit{exactly} the terminal nodes in $\mathcal{T}_\mathsf{top}$, i.e., $C_1,\dots,C_K$.
Consider the forest $\mathcal{F}$ obtained from $\mathcal{T}_\mathsf{top}$ by removing the terminal nodes and the edges incident to them.
The node set of $\mathcal{F}$ is $S$.
The connected components of $\mathcal{F}$ induce a partition $S_1,\dots,S_r$ of $S$ such that there is no edge between $S_i$ and $S_j$ in $\mathcal{T}_\mathsf{top}$ for any $i \neq j$.
Constructing a map $\phi: S \rightarrow \mathbb{R}^2$ is equivalent to constructing maps $\phi_i: S_i \rightarrow \mathbb{R}^2$ for all $i \in [r]$.
Note that the construction of each map $\phi_i$ can be done individually, as there is no edge between $S_i$ and $S_j$ in $\mathcal{T}_\mathsf{top}$ for any $i \neq j$.
Formally, let $\mathcal{T}_i$ be the subtree of $\mathcal{T}_\mathsf{top}$ consisting of the nodes in $S_i$ and their neighbors.
If each $\phi_i$ satisfies $d_{\phi_i}(t,t') \leq 1$ for all $(t,t') \in E(\mathcal{T}_i)$, then the map $\phi$ obtained by setting $\phi_{|S_i} = \phi_i$ satisfies $d_\phi(t,t') \leq 1$ for all $(t,t') \in E(\mathcal{T}_\mathsf{top})$.
Conversely, if $\phi: S \rightarrow \mathbb{R}^2$ satisfies $d_\phi(t,t') \leq 1$ for all $(t,t') \in E(\mathcal{T}_\mathsf{top})$, then $\phi_i = \phi_{|S_i}$ must satisfy $d_{\phi_i}(t,t') \leq 1$ for all $(t,t') \in E(\mathcal{T}_i)$.
Therefore, it suffices to construct each $\phi_i$ individually.
Note that in the tree $\mathcal{T}_i$, every leaf is a terminal node and every internal node is a Steiner point (in $S_i$).
With this reduction, we can assume without loss of generality that the leaves of $\mathcal{T}_\mathsf{top}$ are exactly $C_1,\dots,C_K$.
This assumption guarantees the following nice property of the desired map $\phi$, which turns out to be useful when we design our algorithm.
\begin{lemma} \label{lem-propphi}
Assume the leaves of $\mathcal{T}_\mathsf{top}$ are exactly the terminal nodes in $\mathcal{T}_\mathsf{top}$.
If $\phi: S \rightarrow \mathbb{R}^2$ satisfies $d_\phi(t,t') \leq 1$ for all $(t,t') \in E(\mathcal{T}_\mathsf{top})$, then $\mathsf{dist}(\phi(a),\phi(b)) \leq k-1$ for all $a,b \in S$.
\end{lemma}
\begin{proof}
By the assumption, $S$ is just the set of internal nodes of $\mathcal{T}_\mathsf{top}$, and thus for any $a,b \in S$, the simple path $\pi$ in $\mathcal{T}_\mathsf{top}$ connecting $a$ and $b$ only contains the nodes in $S$.
Suppose $\pi = (t_0,t_1,\dots,t_r)$ where $t_0 = a$ and $t_r = b$.
We have $r \leq k-1$ because $|S| \leq k$.
By definition, $\mathsf{dist}(\phi(t_{i-1}),\phi(t_i)) = d_\phi(t_{i-1},t_i) \leq 1$.
Thus, by triangle inequality, $\mathsf{dist}(\phi(a),\phi(b)) \leq \sum_{i=1}^r \mathsf{dist}(\phi(t_{i-1}),\phi(t_i)) \leq k-1$.
\end{proof}

For convenience, we make $\mathcal{T}_\mathsf{top}$ rooted by picking an arbitrary Steiner point $\mathsf{rt} \in S$ as the root of $\mathcal{T}_\mathsf{top}$.
Our algorithm for constructing $\phi$ borrows the high-level idea of \textit{feasible regions} from the algorithm of Bae et al. \cite{bae2011exact} for the $\ell_1$ variant of EBST.
So before discussing our algorithm, let us first briefly review this idea.
For a node $t \in \mathcal{T}_\mathsf{top}$, we denote by $\mathcal{T}_t$ the subtree of $\mathcal{T}_\mathsf{top}$ rooted at $t$.
For each Steiner point $s \in S$, the \textit{feasible region} $R(s) \subseteq \mathbb{R}^2$ of $s$ is the set of all points $x \in \mathbb{R}^2$ such that there exists a map $\phi_s: S \cap \mathcal{T}_s \rightarrow \mathbb{R}^2$ satisfying $d_{\phi_s}(t,t') \leq 1$ for all $(t,t') \in E(\mathcal{T}_s)$ and $\phi_s(s) = x$.
(Bae et al. \cite{bae2011exact} defined the feasible regions in terms of the $\ell_1$ metric, which is the same except that the Euclidean distance function $\mathsf{dist}$ is replaced with the $\ell_1$ distance function.)
Clearly, the desired map $\phi$ exists iff $R(\mathsf{rt}) \neq \emptyset$.
Furthermore, given $R(s)$ for all $s \in S$, if $R(\mathsf{rt}) \neq \emptyset$, then one can easily obtain the map $\phi$ in a top-down manner as follows.
Arbitrarily pick a point in $R(\mathsf{rt})$ as $\phi(\mathsf{rt})$.
For a non-root node $s \in S$, suppose $\phi(s') \in R(s')$ is already determined for the parent $s'$ of $s$ and we now want to determine $\phi(s)$.
Observe that there exists a point $x \in R(s)$ such that $\mathsf{dist}(\phi(s'),x) \leq 1$.
Indeed, as $\phi(s') \in R(s')$, there exists a map $\phi_{s'}: S \cap \mathcal{T}_t \rightarrow \mathbb{R}^2$ satisfying $d_{\phi_{s'}}(t,t') \leq 1$ for all $(t,t') \in E(\mathcal{T}_{s'})$ and $\phi_{s'}(s') = \phi(s')$.
Set $x = \phi_{s'}(s)$.
By definition, $x \in R(s)$ and $\mathsf{dist}(\phi(s'),x) \leq 1$.
We then define $\phi(s) = x$.
In this way, we can construct the entire map $\phi: S \rightarrow \mathbb{R}^2$ from top to bottom, which satisfies the desired property.

Bae et al. \cite{bae2011exact} showed that for the $\ell_1$ case, $R(s)$ is a rectilinear domain with complexity $n^{O(1)}$ for every $s \in S$, and all these regions can be computed in $n^{O(1)}$ time in a bottom-up fashion.
This directly yields a polynomial-time algorithm to solve the problem with the $\ell_1$ metric.
Unfortunately, in the Euclidean setting, the feasible regions are much more complicated and do not have such nice structures.
As such, we cannot solve the problem by directly computing the feasible regions.
Instead, our algorithm will first guess approximately the locations of the images $\phi(s)$ for $s \in S$ (as well as the locations of the points in $C_1,\dots,C_K$ that connect to the Steiner points) and then defining feasible regions with respect to these approximate locations.
In this way, we can guarantee that the feasible regions are somehow well-behaved, which finally allows us to bound their complexity by exploiting certain geometric properties of the regions as well as techniques from combinatorial geometry.

\subsection{Approximately guessing the locations} \label{sec-guess}
Let $\varGamma$ be a grid in the plane in which each cell is a closed square of side-length $\varepsilon$, where $\varepsilon > 0$ is a sufficiently small constant. For example, one can take $\varepsilon=0.1$. 
For convenience, we also use $\varGamma$ to denote the set of all cells of the grid.
Our algorithm first guesses, for each Steiner point $s \in S$, which cell of $\varGamma$ contains the image $\phi(s)$, and for each terminal node $C_i$, which cell of $\varGamma$ contains the point in $C_i$ within distance 1 from $\phi(s_i)$, where $s_i \in S$ is the parent of $C_i$ in $\mathcal{T}_\mathsf{top}$.
At the first glance, the number of guesses needed is at least $n^{O(k)}$, which we cannot afford.
However, by applying the nice property of $\phi$ in Lemma~\ref{lem-propphi}, we can see that only $(k/\varepsilon)^{O(k)} \cdot n=k^{O(k)} \cdot n$ guesses are sufficient.
Formally, we say a map $\xi: \mathcal{T}_\mathsf{top} \rightarrow \varGamma$ \textit{respects} a map $\phi: S \rightarrow \mathbb{R}^2$ if $\phi(s) \in \xi(s)$ for all $s \in S$ and there exists $x \in C_i \cap \xi(C_i)$ such that $\mathsf{dist}(\phi(s_i),x) \leq 1$ for all $i \in [K]$, where $s_i \in S$ is the parent of $C_i$ in $\mathcal{T}_\mathsf{top}$.
We have the following observation.
\begin{lemma}
There exist $r = k^{O(k)} \cdot n$ maps $\xi_1,\dots,\xi_r: \mathcal{T}_\mathsf{top} \rightarrow \varGamma$ such that if there exists a map $\phi: S \rightarrow \mathbb{R}^2$ satisfying $d_\phi(t,t') \leq 1$ for all $(t,t') \in E(\mathcal{T}_\mathsf{top})$, then there also exists such a map $\phi$ which in addition is respected by $\xi_i$ for some $i \in [r]$.
Furthermore, the maps $\xi_1,\dots,\xi_r$ can be computed in $O(r)$ time given $P$ and $\mathcal{T}_\mathsf{top}$.
\end{lemma}
\begin{proof}
Let $\phi: S \rightarrow \mathbb{R}^2$ be a map satisfying $d_\phi(t,t') \leq 1$ for all $(t,t') \in E(\mathcal{T}_\mathsf{top})$.
Consider the terminal node $C_1 \in \mathcal{T}_\mathsf{top}$.
If $\xi: \mathcal{T}_\mathsf{top} \rightarrow \varGamma$ respects $\phi$, then the cell $\xi(C_1)$ must contain at least one point in $C_1$.
As $|C_1| \leq |P| = n$, there are at most $n$ choices for $\xi(C_1)$.
Suppose now $\xi(C_1) \in \varGamma$ is determined.
Let $o$ be the center of $\xi(C_1)$.
In order to let $\xi$ respect $\phi$, for the parent $s$ of $C_1$ in $\mathcal{T}_\mathsf{top}$, $\phi(s)$ must be within distance $1+\varepsilon$ from $o$.
Furthermore, by Lemma~\ref{lem-propphi}, for any $s'\in S$, 
\begin{equation*}
    \mathsf{dist}(o,\phi(s')) \leq \mathsf{dist}(o,\phi(s)) + \mathsf{dist}(\phi(s),\phi(s')) \leq 1+\varepsilon+(k-1)=k+\varepsilon.
\end{equation*}
Therefore, the $\phi$-images of the Steiner points all lie in the $O((k/\varepsilon)^2)$ cells around $\xi(C_1)$.
It follows that for every Steiner point $s' \in S$, there are $O((k/\varepsilon)^2)$ choices for $\xi(s')$.
Also, for every terminal node $C_i$, the cell $\xi(C_i)$ must be within distance $O(k)$ from $\xi(C_1)$, and thus there are  $O((k/\varepsilon)^2)$ choices for $\xi(C_i)$.
The total number of possibilities of $\xi$ is then $(k/\varepsilon)^{O(K+k)} \cdot n$, which is $k^{O(k)} \cdot n$ since $1/\varepsilon = O(1)$ and $K = O(k)$.
The argument above directly gives an algorithm for computing all possible maps $\xi$ in $k^{O(k)} \cdot n$ time.
\end{proof}

By computing and trying the maps $\xi_1,\dots,\xi_r$ in the above lemma, we can guess a map $\xi: \mathcal{T}_\mathsf{top} \rightarrow \varGamma$ which respects the desired map $\phi$.
The guess of $\xi$ only results in an extra $r$ factor in the running time, which is affordable as $r = k^{O(k)} \cdot n$.
Our problem now becomes computing a map $\phi: S \rightarrow \mathbb{R}^2$ respected by $\xi$ which satisfies $d_\phi(t,t') \leq 1$ for all $(t,t') \in E(\mathcal{T}_\mathsf{top})$.

\subsection{Defining feasible regions} \label{sec-feasible}
Recall the notion of feasible regions discussed before Section~\ref{sec-guess}.
Now we slightly change the definition of feasible regions by defining them with respect to the map $\xi: \mathcal{T}_\mathsf{top} \rightarrow \varGamma$.
For each Steiner point $s \in S$, the \textit{feasible region} $R(s) \subseteq \xi(s)$ of $s$ is the set of all points $x \in \xi(s)$ such that there exists a map $\phi_s: S \cap \mathcal{T}_s \rightarrow \mathbb{R}^2$ respected by $\xi_{|\mathcal{T}_s}$ which satisfies $d_{\phi_s}(t,t') \leq 1$ for all $(t,t') \in E(\mathcal{T}_s)$ and $\phi_s(s) = x$.
For convenience, we also define the feasible region of each terminal node $C_i$ of $\mathcal{T}_\mathsf{top}$ by setting $R(C_i) = C_i \cap \xi(C_i)$.
It is clear that the desired map $\phi$ exists iff $R(\mathsf{rt}) \neq \emptyset$.
Furthermore, given $R(s)$ for all $s \in S$, we can easily recover the map $\phi$, as discussed before Section~\ref{sec-guess}.
So it suffices to show how to compute the feasible regions efficiently.
To this end, we need to first understand how the feasible regions look like.

Recall that the \textit{Minkowski sum} of two regions $A$ and $B$ in $\mathbb{R}^2$ is defined as $A \oplus B = \{a+b: a \in A \text{ and } b \in B\}$.
Let $D = \{(x,y) \in \mathbb{R}^2: x^2 + y^2 \leq 1\}$ be the unit disk centered at the origin.
For each $s \in S$, denote by $\mathsf{Ch}(s)$ the set of children of $s$ in $\mathcal{T}_\mathsf{top}$.
We have the following simple observation.

\begin{lemma} \label{lem-R(s)}
$R(s) = (\bigcap_{s' \in \mathsf{Ch}(s)} (R(s') \oplus D)) \cap \xi(s)$ for all $s \in S$.
\end{lemma}
\begin{proof}
To see $R(s) \subseteq (\bigcap_{s' \in \mathsf{Ch}(s)} (R(s') \oplus D)) \cap \xi(s)$, it suffices to show $R(s) \subseteq R(s') \oplus D$ for all $s' \in \mathsf{Ch}(s)$, because $R(s) \subseteq \xi(s)$ by definition.
Let $x \in R(s)$.
There exists $\phi_s: S \cap \mathcal{T}_s \rightarrow \mathbb{R}^2$ respected by $\xi_{|\mathcal{T}_s}$ which satisfies $d_{\phi_s}(t,t') \leq 1$ for all $(t,t') \in E(\mathcal{T}_s)$ and $\phi_s(s) = x$.
Consider a child $s' \in \mathsf{Ch}(s)$.
If $s' \in S$, then $\phi_s(s') \in R(s')$, because of the map $\phi_{s'}$ obtained by restricting $\phi_s$ to $S \cap \mathcal{T}_{s'}$.
We have $\mathsf{dist}(x,\phi_s(s')) = d_{\phi_s}(s,s') \leq 1$, which implies $x \in R(s') \oplus D$.
If $s' = C_i$ for some $i \in [K]$, then there exists a point $y \in C_i \cap \xi(C_i)$ such that $\mathsf{dist}(x,y) = \mathsf{dist}(\phi_s(s),y) \leq 1$, which also implies $x \in R(s') \oplus D$ since $R(s') = C_i \cap \xi(C_i)$.

To see $R(s) \supseteq (\bigcap_{s' \in \mathsf{Ch}(s)} (R(s') \oplus D)) \cap \xi(s)$, consider a point $x \in (\bigcap_{s' \in \mathsf{Ch}(s)} (R(s') \oplus D)) \cap \xi(s)$.
We construct a map $\phi_s: S \cap \mathcal{T}_s \rightarrow \mathbb{R}^2$ as follows.
Set $\phi_s(s) = x$.
For every $s' \in \mathsf{Ch}(s) \cap S$, there exists a point $y \in R(s')$ such that $\mathsf{dist}(x,y) \leq 1$.
As $y \in R(s')$, there exists $\phi_{s'}: S \cap \mathcal{T}_{s'} \rightarrow \mathbb{R}^2$ respected by $\xi_{|\mathcal{T}_{s'}}$ which satisfies $d_{\phi_{s'}}(t,t') \leq 1$ for all $(t,t') \in E(\mathcal{T}_{s'})$ and $\phi_{s'}(s') = y$.
We then set $\phi_s(t) = \phi_{s'}(t)$ for all $t \in S \cap \mathcal{T}_{s'}$.
It is easy to check from the construction that $\xi_{|\mathcal{T}_s}$ respects $\phi_s$ and $d_{\phi_s}(t,t') \leq 1$ for all $(t,t') \in E(\mathcal{T}_s)$.
Thus, $x \in R(s)$.
\end{proof}

\begin{figure}[!ht]
		\begin{center}
			\includegraphics[width=.45\textwidth]{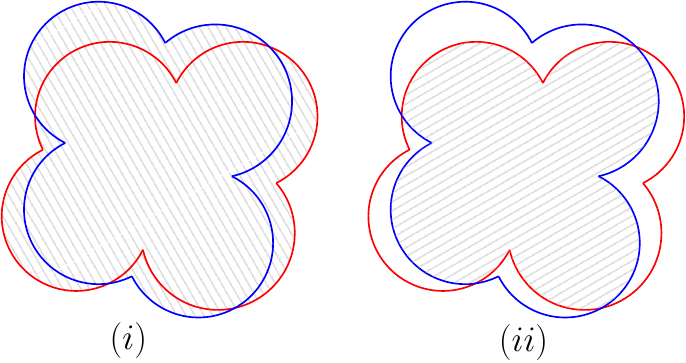}
			%\vspace{-2mm}
			\caption{Illustrating (i) the union and (ii) the intersection of two pseudo-convex circular domains using shaded regions.} 
			\label{fig:1}
		\end{center}
\end{figure}

A \textit{circular arc} refers to a curve that is a connected portion of a circle in $\mathbb{R}^2$.
A segment is also considered as a circular arc (which is a connected portion of an infinitely large circle).
A \textit{circular domain} is a closed subset (or region) of $\mathbb{R}^2$ whose boundary consists of circular arcs that can only intersect at their endpoints. 
Circular domains can be viewed as a generalization of polygons whose boundaries consist of segments and are not self-intersecting. \textit{For our convenience, by abusing notation, we also allow any finite set of points to be a circular domain.}
The \textit{complexity} of any given circular domain $R$, denoted by $\lVert R \rVert$, is defined as follows depending on the structure of $R$. If $R$ is not just a set of points, $\lVert R \rVert$ is the total number of circular arcs and vertices (i.e., the intersection points of the arcs) on its boundary. Otherwise, if $R$ is a set of points, $\lVert R \rVert$ is just the number of these points. 

A circular domain $R$ is called \textit{pseudo-convex} if for every boundary arc $\sigma$ of $R$ that is not a segment, the side of $\sigma$ corresponding to the interior of $R$ coincides with the side of $\sigma$ corresponding to the interior of the disk defining $\sigma$.
Clearly, the union and the intersection of pseudo-convex circular domains are still pseudo-convex circular domains (see Figure \ref{fig:1} for examples).
Furthermore, one can verify that if $R$ is a pseudo-convex circular domain, then the Minkowski sum $R \oplus D$ is also a pseudo-convex circular domain (see Figure \ref{fig:2} for examples).
Therefore, by Lemma~\ref{lem-R(s)}, one can see that $R(s)$ is a pseudo-convex circular domain for every $s \in S$.
For each terminal node $C_i$, $R(C_i)$ is a set of discrete points, which we also view as a pseudo-convex circular domain (with only vertices and no arcs) for convenience.
Using the formula in Lemma~\ref{lem-R(s)} for $R(s)$, one can then compute $R(s)$ in time polynomial in $\sum_{s' \in \mathsf{Ch}(s)} \lVert R(s') \rVert$.
Indeed, each Minkowski sum $R(s') \oplus D$ can be obtained by computing the Minkowski sum of every boundary arc of $R(s')$ and $D$, while the intersection of circular domains can be computed by constructing the arrangement of their boundary arcs. As the center of the circle corresponding to any such arc can be computed in an inductive manner, each arc can be expressed using $O(1)$ complexity. Hence, the intersection points of two such arcs can also be computed in constant time. 
Next, we bound the complexity of the feasible regions.

\begin{figure}[!ht]
		\begin{center}
			\includegraphics[width=.55\textwidth]{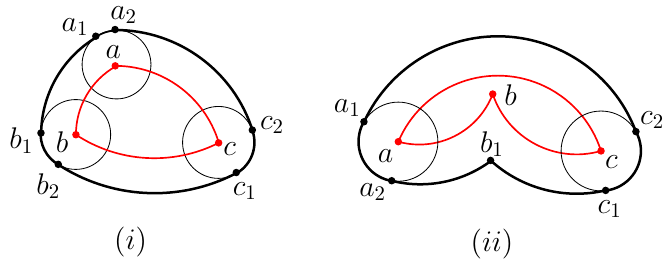}
			%\vspace{-2mm}
			\caption{Illustrating the Minkowski sums of pseudo-convex circular domain $\wideparen{ab},\wideparen{bc},\wideparen{ca}$ and the unit disk $D$. (i) The sum is shown by $\wideparen{a_1b_1}$, $\wideparen{b_1b_2}$, $\wideparen{b_2c_1}$, $\wideparen{c_1c_2}$, $\wideparen{c_2a_2}$, and $\wideparen{a_2a_1}$.  (ii) The sum is shown by $\wideparen{a_1a_2}$, $\wideparen{a_2b_1}$, $\wideparen{b_1c_1}$, $\wideparen{c_1c_2}$, and $\wideparen{c_2a_1}$.} 
			\label{fig:2}
		\end{center}
\end{figure}

\subsection{The complexity of feasible regions} \label{sec-complex}
In this section, we  bound $\lVert R(s) \rVert$ for the Steiner points $s \in S$.
Define the \textit{level} $L(t)$ of a node $t \in \mathcal{T}_\mathsf{top}$ by setting $L(t) = 0$ if $t$ is a leaf and $L(t)=\max_{t' \in \mathsf{Ch}(t)} L(t')+1$ otherwise.
Clearly, $L(s) \geq 1$ for all $s \in S$ and $L(C_i) = 0$ for all $i \in [K]$.
As $|S| \leq k$, the maximum level of a node in $\mathcal{T}_\mathsf{top}$ is bounded by $k$.
By Lemma~\ref{lem-R(s)}, the feasible region of $s \in S$ depends on the feasible regions of its children whose levels are strictly smaller than $L(s)$.
Thus, it is natural to study how the complexity of feasible regions increases along with the level of the nodes.

The nodes of $\mathcal{T}_\mathsf{top}$ at level $0$ are exactly $C_1,\dots,C_K$.
We have $\lVert R(C_i) \rVert = O(n)$ for every $i \in [K]$, as $R(C_i) = C_i \cap \xi(C_i) \subseteq P$ consists of at most $n$ points in the plane.
Assume $\lVert R(t) \rVert \leq m_i$ for all nodes $t \in \mathcal{T}_\mathsf{top}$ with $L(t) \leq i$.
Consider a node $s \in S$ with $L(s) = i+1$.
Now $\lVert R(s') \rVert \leq m_i$ for all $s' \in \mathsf{Ch}(s)$, and we want to bound $\lVert R(s) \rVert$.
As $R(s) = (\bigcap_{s' \in \mathsf{Ch}(s)} (R(s') \oplus D)) \cap \xi(s)$ by Lemma~\ref{lem-R(s)}, the first step is to bound $\lVert R(s') \oplus D \rVert$ for $s' \in \mathsf{Ch}(s)$.
The complexity of Minkowski sums of polygons is well-studied.
Surprisingly, there is not much work in the literature studying Minkowski sums of circular domains.
However, it is not surprising that the insights for Minkowski sums of polygons can also be applied to circular domains.
We show in the following lemma that the known arguments for polygons together with the standard technique of vertical decomposition result in, at least, a linear bound on the complexity of the Minkowski sum of a pseudo-convex circular domain and a unit disk (which is sufficient for our purpose).
% As this part is standard, we defer the proof to Appendix~\ref{appx-proof} due to the limited space.

%Davenport-Schinzel sequence

\begin{lemma} \label{lem-minkowski}
For any pseudo-convex circular domain $R$, we have $\lVert R \oplus D \rVert = O(\lVert R \rVert)$.
\end{lemma}

\begin{proof}
    If $R$ is convex, then one can directly apply the argument for bounding the complexity of the Minkowski sum of two convex polygons (see for example \cite{berg1997computational}) to show that $\lVert R \oplus D \rVert = O(\lVert R \rVert + \lVert D \rVert)$, which is just $O(\lVert R \rVert)$.
In fact, in this case, one can even show that each vertex/arc of $R$ contributes at most one arc of $R \oplus D$, and thus $\lVert R \oplus D \rVert \leq 2 \lVert R \rVert$.

When $R$ is not convex (but pseudo-convex), we shall apply the known argument for bounding the complexity of the Minkowski sum of a non-convex polygon and a convex polygon \cite{kedem1986union}.
Consider the (standard) vertical decomposition of $R$.
Specifically, for each vertex $v$ on the boundary of $R$, we shoot two rays originated from $v$ directing upward and downward, respectively; the rays stop when they touch the boundary of $R$.
Those rays cut $R$ into smaller circular domains $R_1,\dots,R_r$ with total complexity $O(\lVert R \rVert)$.
See Figure~\ref{fig-vertical} for an illustration of this vertical decomposition of $R$.
It is easy to observe that each of these smaller circular domains is convex.
Indeed, each $R_i$ is pseudo-convex, as $R$ is pseudo-convex.
Furthermore, the angle of $R_i$ at every vertex $v$ is at most $\pi$ (here the angle is defined by the two arcs of $R_i$ incident to $v$ with the side corresponding to the interior of $R_i$), for otherwise the angle is cut by one of the two rays originated from $v$ and thus cannot survive in $R_i$.
Therefore, $R_i$ is convex.
It was shown in \cite{kedem1986union} that if $P$ and $Q$ are two interior-disjoint convex regions (in $\mathbb{R}^2$), and $Z$ is another convex region, then $P \oplus Z$ and $Q \oplus Z$ are pseudo-disks, i.e., their boundary cross each other at most twice.
Now $R_1,\dots,R_r$ are interior-disjoint and convex, which implies $R_1 \oplus D,\dots,R_r \oplus D$ are pseudo-disks.
Note that $R \oplus D = \bigcup_{i=1}^r (R_i \oplus D)$.
By the linear union complexity of pseudo-disks~\cite{kedem1986union}, we have $\lVert \bigcup_{i=1}^r (R_i \oplus D) \rVert = O(\sum_{i=1}^r \lVert R_i \oplus D \rVert)$, and the latter is $O(\sum_{i=1}^r \lVert R_i \rVert)$ because each $R_i$ is convex.
Finally, as the total complexity of $R_1,\dots,R_r$ is $O(\lVert R \rVert)$, we have $\lVert R \oplus D \rVert = O(\lVert R \rVert)$.
\end{proof}

\begin{figure}
    \centering
    \includegraphics[height=4.7cm]{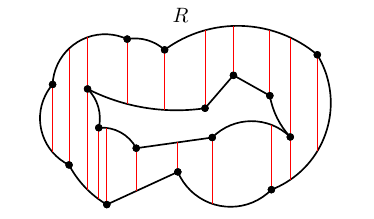}
    \caption{Illustrating the vertical decomposition of $R$.}
    \label{fig-vertical}
\end{figure}

Applying the above lemma, we have $\lVert R(s') \oplus D \rVert = O(m_i)$ for all $s' \in \mathsf{Ch}(s)$.
Next, we are going to bound the complexity of the intersection $\bigcap_{s' \in \mathsf{Ch}(s)} (R(s') \oplus D)$.
This step is more challenging, and is achieved by proving the following key lemma.

\begin{lemma} \label{lem-key}
Let $R_1,\dots,R_r$ be circular domains each of which is inside an $\varepsilon \times \varepsilon$ grid cell for $\varepsilon \leq 0.1$.
Then $\lVert \bigcap_{i=1}^r (R_i \oplus D) \rVert = O(r^2 M)$, where $M = \max_{i=1}^r \lVert R_i \oplus D \rVert$.
\end{lemma}

As the proof of Lemma~\ref{lem-key} is technical, we defer it to Section~\ref{sec-proof}.
At this point, let us first finish the discussion assuming this lemma.
By Lemma~\ref{lem-key}, we have $\lVert \bigcap_{s' \in \mathsf{Ch}(s)} (R(s') \oplus D) \rVert = O(r^2 m_i)$, where $r = |\mathsf{Ch}(s)|$.
Recall that each Steiner point in $\mathcal{T}_\mathsf{top}$ is of degree at most 5, which implies $r \leq 5$.
Therefore, $\lVert \bigcap_{s' \in \mathsf{Ch}(s)} (R(s') \oplus D) \rVert = O(m_i)$.
To further bound $\lVert R(s) \rVert$ is easy.
By Lemma~\ref{lem-R(s)}, $R(s)$ is the intersection of $\bigcap_{s' \in \mathsf{Ch}(s)} (R(s') \oplus D)$ and the grid cell $\xi(s)$.
Note that taking intersection with a square can only increase the complexity by a constant factor.
Indeed, $\xi(s)$ has four edges and each edge intersects each boundary arc of $\bigcap_{s' \in \mathsf{Ch}(s)} (R(s') \oplus D)$ at most twice.
Therefore, $\lVert R(s) \rVert = O(m_i)$.
In other words, from level $i$ to level $i+1$, the complexity of the feasible regions only increase by a constant factor.
It follows that $\lVert R(s) \rVert = 2^{O(j)} \cdot n$ for any node $s \in \mathcal{T}_\mathsf{top}$ with $L(s) = j$.
Since $\mathcal{T}_\mathsf{top}$ has at most $k$ levels, we conclude that $\lVert R(s) \rVert = 2^{O(k)} \cdot n$ for all $s \in \mathcal{T}_\mathsf{top}$.

\subsubsection{Proof of Lemma~\ref{lem-key}} \label{sec-proof}

In this section, we prove Lemma~\ref{lem-key}.
Let $R_1,\dots,R_r$ be the circular domains in the lemma.
We prove that each $R_i \oplus D$ is a \textit{simple} circular domain (i.e., connected and without holes).
%, that is, the segment $\overline{op}$ is contained in $R \oplus D$ for any $p \in D$.
In fact, we prove an even stronger property of $R_i \oplus D$.

\begin{observation} \label{obs-star}
Let $\Box$ be an $\varepsilon \times \varepsilon$ grid cell and $R \subseteq \Box$ be any set.
For any points $o \in \Box$ and $p \in R \oplus D$, the segments $\overline{op}$ is contained in $R \oplus D$, and furthermore the interior of $\overline{op}$ is contained in the interior of $R \oplus D$.
%In particular, $R \oplus D$ is a simple circular domain.
\end{observation}
\begin{proof}
Since $p \in R \oplus D$, there exists $q \in R \subseteq \Box$ such that $p \in D_q$, where $D_q$ is the unit disk centered at $q$.
As $o,q \in \Box$ and the side-length of $\Box$ is $\varepsilon \leq 0.1$ by the assumption in Lemma~\ref{lem-key}, we have $o \in D_q$.
By the convexity of $D_q$, we have $\overline{op} \subseteq D_q$.
Note that $D_q \subseteq R \oplus D$ because $q \in R$, which implies $\overline{op} \subseteq R \oplus D$.
Furthermore, the interior of $\overline{op}$ is in the interior of $D_q$ and thus in the interior of $R \oplus D$.
\end{proof}

The above observation implies that each $R_i \oplus D$ is star-shaped (i.e., there exists one point $o \in R_i \oplus D$ that sees the entire region $R_i \oplus D$), and in particular, a simple circular domain.
We use $\partial (R_i \oplus D)$ to denote the boundary of $R_i \oplus D$, which is a simple closed circular curve (i.e., curve consisting of circular arcs).
Consider the arrangement $\mathcal{A}$ of all circular curves $\partial (R_1 \oplus D),\dots,\partial (R_r \oplus D)$.
The vertices of $\mathcal{A}$ are the vertices of the curves $\partial (R_1 \oplus D),\dots,\partial (R_r \oplus D)$ and their proper intersection points\footnote{A point $p \in \mathbb{R}^2$ is a \textit{proper intersection point} of two curves $\gamma$ and $\gamma'$ in $\mathbb{R}^2$, if $p \in \gamma \cap \gamma'$ and $p$ is \textit{isolated} in $\gamma \cap \gamma'$, i.e., there is a open neighborhood $U \subseteq \mathbb{R}^2$ of $p$ such that $U \cap (\gamma \cap \gamma') = \{p\}$. It can happen that two curves $\partial (R_i \oplus D)$ and $\partial (R_j \oplus D)$ have infinitely many intersection points when an arc of $\partial (R_i \oplus D)$ overlaps with an arc of $\partial (R_j \oplus D)$. But these are not proper intersection points and do not contribute vertices of $\mathcal{A}$.}.
These vertices subdivide the curves $\partial (R_1 \oplus D),\dots,\partial (R_r \oplus D)$ into smaller pieces each of which is a circular arc; they are the edges of $\mathcal{A}$.
Finally, the curves $\partial (R_1 \oplus D),\dots,\partial (R_r \oplus D)$ subdivide the plane $\mathbb{R}^2$ into connected regions, which are the faces of $\mathcal{A}$.
Clearly, $\bigcap_{i=1}^r (R_i \oplus D)$ is the union of several faces of $\mathcal{A}$.
Therefore, $\lVert \bigcap_{i=1}^r (R_i \oplus D) \rVert$ is bounded by the total number of vertices and edges of $\mathcal{A}$.
%Note that the number of edges of $\mathcal{A}$ is at most $2r$ times the number of vertices of $\mathcal{A}$, as each vertex can be incident to at most $2r$ edges (each of the $r$ curves contributes at most two edges).
We first consider the number of vertices of $\mathcal{A}$ (bounding the number of edges is easy once we know the number of vertices).

The total number of vertices of the curves $\partial (R_1 \oplus D),\dots,\partial (R_r \oplus D)$ is bounded by $O(rM)$ where $M = \max_{i=1}^r \lVert R_i \oplus D \rVert$.
Besides these points, the other vertices of $\mathcal{A}$ are proper intersection points of $\partial (R_1 \oplus D),\dots,\partial (R_r \oplus D)$.
In what follows, we shall prove that the number of proper intersection points of $\partial (R_i \oplus D)$ and $\partial (R_j \oplus D)$ is bounded by $O(M)$, for any $i,j \in [r]$.
Note that this bound does not hold for \textit{general} circular curves.
Indeed, one can easily show that two circular curves of complexity $M$ can have $\Theta(M^2)$ proper intersection points in worst case (even if both curves are boundaries of star-shaped pseudo-convex circular domains).
Our proof heavily relies on the fact that the two curves are both the boundaries of the Minkowski sums of a ``tiny'' circular domain and the unit disk $D$.
More specifically, the geometry of such Minkowski sums allows us to somehow relate the two curves with \textit{monotone} curves, which are well-behaved.
Without loss of generality, it suffices to consider the curves $\partial (R_1 \oplus D)$ and $\partial (R_2 \oplus D)$, i.e., the case $i=1$ and $j=2$.
%As such, we need to first establish some geometric properties of such Minkowski sums.

Recall that a curve $\gamma$ in $\mathbb{R}^2$ is \textit{$x$-monotone} if there exists a homeomorphism $f:[0,1] \rightarrow \gamma$ such that the $x$-coordinate of $f(a)$ is smaller than the $x$-coordinate of $f(b)$ whenever $a < b$.
Intuitively, $\gamma$ is $x$-monotone if we always move right (or left) when going along $\gamma$ from one side to the other side.
By slightly generalizing this notion, we can define the monotonicity of a curve with respect to a vector.
Formally, let $\vec{v} \in \mathbb{S}^1$ be a unit vector.
We say a curve $\gamma$ in $\mathbb{R}^2$ is \textit{$\vec{v}$-monotone} if there exists a homeomorphism $f:[0,1] \rightarrow \gamma$ such that $\langle \vec{v},f(a) \rangle < \langle \vec{v},f(b) \rangle$ whenever $a < b$; here $\langle \cdot, \cdot \rangle$ denotes the inner product.
Two curves $\gamma$ and $\gamma'$ in $\mathbb{R}^2$ are \textit{mutually monotone} if there exists $\vec{v} \in \mathbb{S}^1$ such that $\gamma$ and $\gamma'$ are both $\vec{v}$-monotone.
One can easily show that two mutually monotone curves have linear number of proper intersection points.

\begin{fact} \label{fact-monoint}
Let $\gamma$ and $\gamma'$ be two circular curves that are mutually monotone.
Then $\gamma$ and $\gamma'$ have $O(m+m')$ proper intersection points, where $m = \lVert \gamma \rVert$ and $m' = \lVert \gamma' \rVert$.
\end{fact}
\begin{proof}
Without loss of generality, assume $\gamma$ and $\gamma'$ are $x$-monotone.
Let $V$ (resp., $V'$) denote the set of vertices on $\gamma$ (resp., $\gamma'$).
Suppose the $x$-coordinates of the points in $V \cup V'$ are $x_0,x_1,\dots,x_r$, where $r = O(m+m')$ and $x_0 \leq x_1 \leq \cdots \leq x_r$.
For each $i \in [r]$, define $X_i = [x_{i-1},x_i] \times \mathbb{R}^2$, which is a vertical strip.
The part of $\gamma$ (resp., $\gamma'$) inside each $X_i$ is a circular arc.
Two circular arcs can have at most two proper intersection points.
As such, $\gamma$ and $\gamma'$ have at most two proper intersection points inside each $X_i$, and thus have in total at most $2r = O(m+m')$ proper intersection points.
\end{proof}

Based on the above observation, our key idea is to decompose $\partial (R_1 \oplus D)$ and $\partial (R_2 \oplus D)$ into $O(1)$ pieces such that any two of these pieces are mutually monotone.
As long as this is possible, we can easily bound the number of proper intersection points of $\partial (R_1 \oplus D)$ and $\partial (R_2 \oplus D)$.
To do this decomposition, we need to first establish some nice geometric properties of the Minkowski sum of a ``tiny'' circular domain and the unit disk $D$.

Fix a circular domain $R$ that is inside an $\varepsilon \times \varepsilon$ grid cell $\Box$, and let $o$ be the center of $\Box$.
%One can easily see that $R \oplus D$ is a star-convex region centered at $o$, that is, the segment $\overline{op}$ is contained in $R \oplus D$ for any $p \in D$.
%In fact, $R \oplus D$ even admits a slightly stronger property.
Observation~\ref{obs-star} implies that any ray $r$ originated from $o$ intersects $\partial (R \oplus D)$ at exactly one point.
Indeed, if $r$ intersects $\partial (R \oplus D)$ at two points $a$ and $b$ where $a$ is closer to $o$ than $b$, then the interior of $\overline{ob}$ is not contained in the interior of $R \oplus D$ as ($a$ is in the interior of $\overline{ob}$ but not in the interior of $R \oplus D$), which contradicts Observation~\ref{obs-star}.
It follows that the map $\pi: \partial (R \oplus D) \rightarrow \mathbb{S}^1$ defined as $\pi(p) = \overrightarrow{op}/ \lVert \overrightarrow{op} \rVert$ is bijective and is thus a homeomorphism.
%In particular, $R \oplus D$ is a \textit{simple} circular domain (i.e., connected and without holes).
Therefore, $\partial (R \oplus D)$ has the nice property that if a point $p$ moves along $\partial (R \oplus D)$ in one direction, then the vector $\overrightarrow{op}$ also rotates in one direction.

A point $p \in \partial (R \oplus D)$ is a \textit{non-vertex point} if it is in the interior of a circular arc of $\partial (R \oplus D)$.
Consider a non-vertex point $p \in \partial (R \oplus D)$ lying in the interior of the circular arc $\sigma$ of $\partial (R \oplus D)$.
There exists a unique line $\ell$ going through $p$ that is tangent to $\partial (R \oplus D)$, which is also the tangent line of $\sigma$ at $p$.
A \textit{tangent vector} of $R \oplus D$ at $p$ refers to a unit vector parallel to $\ell$.
Clearly, there are two tangent vectors of $R \oplus D$ at $p$, among which one vector $\vec{v}$ indicates the clockwise direction in the sense that if $p$ moves along $\partial (R \oplus D)$ in the direction of $\vec{v}$, then $\overrightarrow{op}$ rotates clockwise.
We then call $\vec{v}$ the \textit{clockwise tangent vector} (or \textit{clockwise tangent} for short) of $R \oplus D$ at $p$, and denote it by $\mathsf{tan}_p$.
In the next observation, we prove that $\mathsf{tan}_p$ is almost perpendicular to the vector $\overrightarrow{op}$.
For two nonzero vectors $\vec{u}$ and $\vec{v}$ in the plane, let $\mathsf{ang}(\vec{u},\vec{v})$ denote the \textit{clockwise ordered angle} from $\vec{u}$ to $\vec{v}$, i.e., the angle between $\vec{u}$ and $\vec{v}$ that is to the clockwise of $\vec{u}$ and to the counter-clockwise of $\vec{v}$.

\begin{observation} \label{obs-almostperp}
%Let $R$ be a circular domain inside an $\varepsilon \times \varepsilon$ grid cell $\Box$ and $o \in \Box$ be the center of $\Box$.
For any non-vertex point $p \in \partial (R \oplus D)$, we have $|\mathsf{ang}(\overrightarrow{op},\mathsf{tan}_p) - \frac{\pi}{2}| \leq 4 \varepsilon$, where $\mathsf{tan}_p$ is the clockwise tangent of $R \oplus D$ at $p$.
\end{observation}
\begin{proof}
Consider a non-vertex point $p \in \partial (R \oplus D)$ and suppose it is in the interior of a circular arc $\sigma$ of $\partial (R \oplus D)$.
As $p \in \partial (R \oplus D)$, there exists a point $q \in R$ such that $\mathsf{dist}(p,q) = 1$.
Let $D_q$ be the unit disk centered at $q$, whose boundary $\partial D_q$ contains $p$.
Now the circular arc $\sigma$ and the circle $\partial D_q$ intersect at $p$.
But $\sigma$ and $\partial D_q$ cannot cross each other at $p$, because $\sigma$ is a portion of $\partial (R \oplus D)$ and $D_q \subseteq R \oplus D$.
Therefore, $\sigma$ and $\partial D_q$ are tangent at $p$, and they share the same tangent line $\ell$ at $p$.
See Figure~\ref{fig-almostperp} for an illustration.
As a tangent line of $D_q$ at $p$, $\ell$ is perpendicular to $\overrightarrow{qp}$.
We claim that the (smaller) angle between $\ell$ and $\overrightarrow{op}$ is at least $\frac{\pi}{2} - 4 \varepsilon$.
It suffices to show that $\angle opq \leq 4 \varepsilon$, since the angle between $\ell$ and $\overrightarrow{qp}$ is $\frac{\pi}{2}$.
We have $\mathsf{dist}(p,q) = 1$ and $\mathsf{dist}(o,q) \leq 2 \varepsilon$.
By the formula $\frac{\mathsf{dist}(o,q)}{\mathsf{dist}(p,q)} = \frac{\sin \angle opq}{\sin \angle poq}$, we have that $\frac{\sin \angle opq}{\sin \angle poq} \leq 2 \varepsilon$ and thus $\sin \angle opq \leq 2 \varepsilon$.
It follows that $\angle opq \leq 4 \varepsilon$.
Let $\mathsf{tan}_p$ be the clockwise tangent of $R \oplus D$ at $p$.
Since the angle between $\ell$ and $\overrightarrow{op}$ is at least $\frac{\pi}{2} - 4 \varepsilon$, we have either $|\mathsf{ang}(\overrightarrow{op},\mathsf{tan}_p) - \frac{\pi}{2}| \leq 4 \varepsilon$ or $|\mathsf{ang}(\overrightarrow{op},\mathsf{tan}_p) - \frac{3\pi}{2}| \leq 4 \varepsilon$.
But by the definition of \textit{clockwise} tangent, $\mathsf{ang}(\overrightarrow{op},\mathsf{tan}_p) \leq \pi$.
Therefore, $|\mathsf{ang}(\overrightarrow{op},\mathsf{tan}_p) - \frac{\pi}{2}| \leq 4 \varepsilon$.
%The clockwise tangent $\vec{v}$ of $R \oplus D$ at $p$ is also a tangent vector of $D_q$ at $p$, and is thus perpendicular to $\overrightarrow{qp}$.
\end{proof}

\begin{figure}[h]
    \centering
    \includegraphics[height=4.5cm]{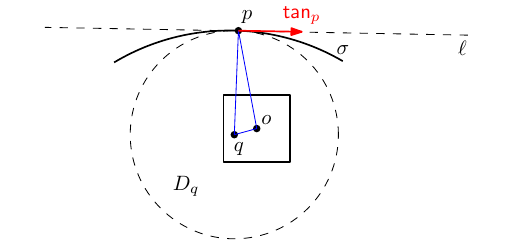}
    \caption{Illustrating the proof of Observation~\ref{obs-almostperp}.}
    \label{fig-almostperp}
\end{figure}

The next observation allows us to test the monotoncity of a piece of $\partial (R \oplus D)$ by checking the clockwise tangent vectors at points on that piece.

\begin{observation} \label{obs-monotone}
%Let $R$ be a circular domain inside an $\varepsilon \times \varepsilon$ grid cell $\Box$ and $o \in \Box$ be the center of $\Box$.
Let $\gamma$ be a connected portion of $\partial (R \oplus D)$ and $\vec{v} \in \mathbb{S}^1$ be a unit vector.
If $\langle \vec{v},\mathsf{tan}_p \rangle > 0$ for any non-vertex point $p \in \gamma$ or $\langle \vec{v},\mathsf{tan}_p \rangle < 0$ for any non-vertex point $p \in \gamma$, then $\gamma$ is $\vec{v}$-monotone.
\end{observation}
\begin{proof}
Without loss of generality, we only need to consider the case $\langle \vec{v},\mathsf{tan}_p \rangle > 0$ for any non-vertex point $p \in \gamma$.
Let $a,b \in \gamma$ be the two endpoints of $\gamma$ such that $\gamma$ is on the clockwise (resp., counterclockwise) side of $a$ (resp., $b$).
Consider a homeomorphism $f:[0,1] \rightarrow \gamma$ with $f(0) = a$ and $f(1) = b$.
Define $I = \{x \in [0,1]: f(x) \text{ is a vertex point}\}$ and suppose $I = \{x_1,\dots,x_{r-1}\}$ where $x_1 < \cdots < x_{r-1}$.
Set $x_0 = 0$ and $x_r = 1$.
For each $i \in [r]$, the image $f([x_{i-1},x_i])$ is a circular arc on $\partial (R \oplus D)$.
When $x$ goes from $x_{i-1}$ to $x_i$, $f(x)$ is moving clockwise around $o$.
As $\langle \vec{v},\mathsf{tan}_{f(x)} \rangle > 0$ for any $x \in (x_{i-1},x_i)$ and $f(x)$ is moving clockwise around $o$ when $x$ goes from $x_{i-1}$ to $x_i$, the function $g(x) = \langle \vec{v},f(x) \rangle$ is increasing on the open interval $(x_{i-1},x_i)$ and is thus increasing on the closed interval $[x_{i-1},x_i]$ because it is continuous.
Since $[0,1] = \bigcup_{i=1}^r [x_{i-1},x_i]$, $g(x)$ is increasing on the entire range $[0,1]$, which implies that $\gamma$ is $\vec{v}$-monotone.
\end{proof}

%We subdivide $\partial (R \oplus D)$ into multiple circular curves as follows.
Let $\Lambda \geq 10$ be an integer.
Consider a subdivision of $\partial (R \oplus D)$ into $\Lambda$ pieces as follows.
We shoot $\Lambda$ rays from $o$, which evenly divide the $2\pi$ angle around $o$ into $\Lambda$ angles each of size $2\pi/\Lambda$.
These rays intersect $\partial (R \oplus D)$ at $\Lambda$ points $p_1,\dots,p_\Lambda$ such that $\mathsf{ang}(\overrightarrow{op_{i-1}},\overrightarrow{op_i}) = 2\pi/\Lambda$; for convenience, here we set $p_0 = p_\Lambda$.
The points $p_1,\dots,p_\Lambda$ subdivide $\partial (R \oplus D)$ into $\Lambda$ circular curves $\gamma_1,\dots,\gamma_\Lambda$, where $\gamma_i$ is the one with endpoints $p_{i-1}$ and $p_i$.
We call $\gamma_1,\dots,\gamma_\Lambda$ a \textit{$\Lambda$-decomposition} of $\partial (R \oplus D)$.
See Figure~\ref{fig-decomp} for an illustration (while we require $\Lambda \geq 10$, the figure only shows $\Lambda = 8$ for simplicity).
We observe that each curve $\gamma_i$ is monotone with respect to almost all vectors $\vec{v} \in \mathbb{S}^1$ (if $\Lambda$ is sufficiently large and $\varepsilon$ is sufficiently small).

\begin{figure}[h]
    \centering
    \includegraphics[height=5.2cm]{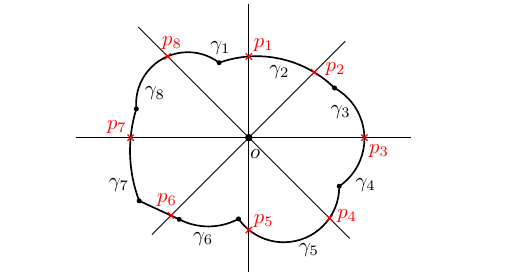}
    \caption{Illustrating the $\Lambda$-decomposition for $\Lambda = 8$.}
    \label{fig-decomp}
\end{figure}

\begin{observation} \label{obs-monorange}
%Suppose $\Lambda \geq 10$ and $\varepsilon \leq 0.1$.
For every $i \in [\Lambda]$, $\gamma_i$ is $\vec{v}$-monotone for all $\vec{v} \in \mathbb{S}^1$ such that $\varepsilon' \leq \mathsf{ang}(\overrightarrow{op_i},\vec{v}) \leq \pi - \varepsilon'$ or $\pi + \varepsilon' \leq \mathsf{ang}(p_i,\vec{v}) \leq 2\pi - \varepsilon'$, where $\varepsilon' = 7(\frac{1}{\Lambda}+\varepsilon)$.
\end{observation}
\begin{proof}
Let $\vec{v} \in \mathbb{S}^1$.
Without loss of generality, we only need to consider the case $\varepsilon' \leq \mathsf{ang}(\overrightarrow{op_i},\vec{v}) \leq \pi - \varepsilon'$ because $\gamma_i$ is $\vec{v}$-monotone iff $\gamma_i$ is $(-\vec{v})$-monotone.
By construction, we have $\mathsf{ang}(\overrightarrow{op},\overrightarrow{op_i}) \leq 2\pi/\Lambda$ for any non-vertex point $p \in \gamma_i$.
Furthermore, by Observation~\ref{obs-almostperp}, we have $|\mathsf{ang}(\overrightarrow{op},\mathsf{tan}_p) - \frac{\pi}{2}| \leq 4\varepsilon$ for any non-vertex point $p \in \gamma_i$.
As $\Lambda \geq 10$ and $\varepsilon \leq 0.1$, both $2\pi/\Lambda$ and $4 \varepsilon$ are sufficiently small.
So we have $|\mathsf{ang}(\overrightarrow{op_i},\mathsf{tan}_p) - \frac{\pi}{2}| \leq 2\pi/\Lambda + 4\varepsilon$ for any non-vertex point $p \in \gamma_i$.
We claim that $\langle \vec{v},\mathsf{tan}_p \rangle > 0$ for any non-vertex point $p \in \gamma_i$.
Let $\vec{u} \in \mathbb{S}^1$ be the unique vector satisfying $\mathsf{ang}(\overrightarrow{op_i},\vec{u}) = \frac{\pi}{2}$.
Then the smaller angle between $\vec{v}$ and $\vec{u}$ is at most $\frac{\pi}{2} - \varepsilon'$.
By the fact $|\mathsf{ang}(\overrightarrow{op_i},\mathsf{tan}_p) - \frac{\pi}{2}| \leq 2\pi/\Lambda + 4\varepsilon$, we know that the smaller angle between $\mathsf{tan}_p$ and $\vec{u}$ is at most $2\pi/\Lambda + 4\varepsilon$, which is strictly smaller than $7(\frac{1}{\Lambda}+\varepsilon) = \varepsilon'$.
Therefore, the smaller angle between $\vec{v}$ and $\mathsf{tan}_p$ is strictly smaller than $\frac{\pi}{2}$, which implies $\langle \vec{v},\mathsf{tan}_p \rangle > 0$.
Finally, by Observation~\ref{obs-monotone}, $\gamma_i$ is $\vec{v}$-monotone.
\end{proof}

%It follows that $R \oplus D$ is a \textit{simple} circular domain.
Now we are ready to bound the number of proper intersection points of $\partial (R_1 \oplus D)$ and $\partial (R_2 \oplus D)$.
Take a $\Lambda$-decomposition $\gamma_1,\dots,\gamma_\Lambda$ of $\partial (R_1 \oplus D)$ and a $\Lambda$-decomposition $\eta_1,\dots,\eta_\Lambda$ of $\partial (R_2 \oplus D)$ for $\Lambda = 10$.
Observe that $\gamma_i$ and $\eta_j$ are mutually monotone for all $i,j \in [\Lambda]$.
To see this, let $V_\gamma = \{\vec{v} \in \mathbb{S}^1: \gamma_i \text{ is } \vec{v}\text{-monotone}\}$ and $V_\eta = \{\vec{v} \in \mathbb{S}^1: \eta_j \text{ is } \vec{v}\text{-monotone}\}$.
By Observation~\ref{obs-monorange}, $V_\gamma$ contains two disjoint arcs on $\mathbb{S}^1$ each of length $\pi - 2 \varepsilon'$, and so does $V_\eta$.
As $\Lambda = 10$ and $\varepsilon \leq 0.1$, we have $2 \varepsilon' \leq 1.4 \leq \frac{\pi}{2}$ and thus $\pi - 2 \varepsilon' \geq \frac{\pi}{2}$.
Therefore, $V_\gamma$ (resp., $V_\eta$) contains two disjoint arcs on $\mathbb{S}^1$ each of length at least $\frac{\pi}{2}$.
It follows that $V_\gamma \cap V_\eta \neq \emptyset$, which further implies that $\gamma_i$ and $\eta_j$ are mutually monotone, as both of them are $\vec{v}$-monotone for any $\vec{v} \in V_\gamma \cap V_\eta$.
By Fact~\ref{fact-monoint}, the number proper intersection points of $\gamma_i$ and $\eta_j$ is $O(M)$.
As $\Lambda = 10$, we know that $\partial (R_1 \oplus D)$ and $\partial (R_2 \oplus D)$ have $O(M)$ proper intersection points.

As the number of proper intersection points of any two curves $\partial (R_i \oplus D)$ and $\partial (R_j \oplus D)$ is $O(M)$, the total number of vertices of the arrangement $\mathcal{A}$ is $O(r^2 M)$. By the definition of $\mathcal{A}$, it may contain overlapping edges. Thus, it is not exactly a planar graph. However, it is not hard to see that the number of edges of $\mathcal{A}$ is $O(r^3 M)$.
Indeed, the number of edges of $\mathcal{A}$ is at most $2r$ times the number of vertices of $\mathcal{A}$, as each vertex can be incident to at most $2r$ edges (each of the $r$ curves contributes at most two edges).
In fact, we can improve the bound to $O(r^2 M)$ as follows.
For each vertex $v$ of $\mathcal{A}$, define $\Delta_v = \Delta_v' + \Delta_v''$, where $\Delta_v'$ is the number of indices $i \in [r]$ such that $v$ is a vertex of $\partial (R_i \oplus D)$ and $\Delta_v''$ is the number of pairs $(i,j) \in [r] \times [r]$ such that $v$ is a proper intersection point of $\partial (R_i \oplus D)$ and $\partial (R_j \oplus D)$.
We claim that $\sum_v \Delta_v' = O(rM)$ and $\sum_v \Delta_v'' = O(r^2 M)$.
Indeed, each vertex of each curve $\partial (R_i \oplus D)$ is counted once in $\sum_v \Delta_v'$ and each proper intersection point of each pair of curves is counted once in $\sum_v \Delta_v''$.
As each curve $\partial (R_i \oplus D)$ has $O(M)$ vertices and each pair of curves have $O(M)$ proper intersection points as shown before, we have $\sum_v \Delta_v' = O(rM)$ and $\sum_v \Delta_v'' = O(r^2 M)$, which implies $\sum_v \Delta_v = O(r^2 M)$.
Clearly, the number of edges of $\mathcal{A}$ incident to each vertex $v$ is bounded by $O(\Delta_v)$.
Therefore, the total number of edges of $\mathcal{A}$ is $O(r^2 M)$.
Finally, we conclude that $\lVert \bigcap_{i=1}^r (R_i \oplus D) \rVert = O(r^2 M)$.

\subsection{Putting everything together}

In Section~\ref{sec-complex}, we obtain the bound $\lVert R(s) \rVert = 2^{O(k)} \cdot n$ for all $s \in \mathcal{T}_\mathsf{top}$.
Combining it with the discussion in Section~\ref{sec-feasible}, we can compute a map $\phi: S \rightarrow \mathbb{R}^2$ respected by a given $\xi: \mathcal{T}_\mathsf{top} \rightarrow \varGamma$ satisfying $d_\phi(t,t')$ for all $(t,t') \in E(\mathcal{T}_\mathsf{top})$ or decide the non-existence of such a map in $2^{O(k)} \cdot n^{O(1)}$ time.
Further combining this with our guess for $\xi$ in Section~\ref{sec-guess}, we obtain a $k^{O(k)} n^{O(1)}$-time algorithm for the decision version of fixed-topology EBST.

Recall the standard parametric search technique of Megiddo~\cite{megiddo1983applying}, which can convert a decision algorithm $\mathcal{A}$ for some problem with running time $T(n)$ to an optimization algorithm for the same problem with running time $O(T^2(n))$, by using $\mathcal{A}$ as both the test algorithm and the decision algorithm.
Therefore, our $k^{O(k)} n^{O(1)}$-time decision algorithm implies a $k^{O(k)} n^{O(1)}$-time optimization algorithm algorithm (for fixed-topology EBST) as well.
%Moreover, as we discuss in the following subsection, one can transform the $k^{O(k)} n^{O(1)}$-time decision algorithm to an optimization algorithm (for fixed-topology EBST) whose time complexity is asymptotically same as the running time of the decision algorithm, which is still $k^{O(k)} \cdot n^{O(1)}$.
Finally, combining this with the reduction in Section~\ref{sec-reduction}, we obtain a $k^{O(k)} n^{O(1)}$-time algorithm for EBST.

\main*

One can verify that our proof of Theorem~\ref{thm:ebst} also works for the $\ell_p$ variant of EBST for any rational $1 \leq p \leq \infty$.
In the $\ell_p$ setting, everything is the same except that the unit disk $D$ becomes an $\ell_p$ unit disk.
Consequently, the boundary arcs of the ``circular'' domains become $\ell_p$ arcs.
However, our arguments did not use any special property of \textit{circular} arcs, and thus still apply (possibly with different parameters $\varepsilon$ and $\Lambda$).
More generally, our algorithm works as long as $D$ is convex and its boundary consists of $O(1)$ algebraic curves. Due to the $O(1)$ complexity of these curves, the intersection points of any two of them can be computed in constant time. 
As such, Theorem~\ref{thm:ebst} can also be extended to some other metrics on $\mathbb{R}^2$, such as a positive linear combination of $\ell_p$ metrics.

\bibliographystyle{plainurl}
\bibliography{BST,steiner}

\begin{thebibliography}{10}

\bibitem{abu2011euclidean}
A~Karim Abu-Affash.
\newblock On the euclidean bottleneck full steiner tree problem.
\newblock In {\em Proceedings of the twenty-seventh annual symposium on
  Computational geometry}, pages 433--439, 2011.

\bibitem{abu2019bottleneck}
A~Karim Abu-Affash, Sujoy Bhore, Paz Carmi, and Dibyayan Chakraborty.
\newblock Bottleneck bichromatic full steiner trees.
\newblock {\em Information Processing Letters}, 142:14--19, 2019.

\bibitem{abu2011path}
A~Karim Abu-Affash, Paz Carmi, Matthew~J Katz, and Michael Segal.
\newblock The euclidean bottleneck steiner path problem.
\newblock In {\em Proceedings of the twenty-seventh annual symposium on
  Computational geometry}, pages 440--447, 2011.

\bibitem{arora1998polynomial}
Sanjeev Arora.
\newblock Polynomial time approximation schemes for euclidean traveling
  salesman and other geometric problems.
\newblock {\em Journal of the ACM (JACM)}, 45(5):753--782, 1998.

\bibitem{bae2011exact}
Sang~Won Bae, Sunghee Choi, Chunseok Lee, and Shin-ichi Tanigawa.
\newblock Exact algorithms for the bottleneck steiner tree problem.
\newblock {\em Algorithmica}, 61(4):924--948, 2011.

\bibitem{bae2010exact}
Sang~Won Bae, Chunseok Lee, and Sunghee Choi.
\newblock On exact solutions to the euclidean bottleneck steiner tree problem.
\newblock {\em Information Processing Letters}, 110(16):672--678, 2010.

\bibitem{BasavarajuFGMRS14}
Manu Basavaraju, Fedor~V. Fomin, Petr~A. Golovach, Pranabendu Misra, M.~S.
  Ramanujan, and Saket Saurabh.
\newblock Parameterized algorithms to preserve connectivity.
\newblock In {\em Proceedings of the 41st International Colloquium of Automata,
  Languages and Programming (ICALP)}, volume 8572 of {\em Lecture Notes in
  Comput. Sci.}, pages 800--811. Springer, 2014.

\bibitem{berg1997computational}
Mark~de Berg, Marc~van Kreveld, Mark Overmars, and Otfried Schwarzkopf.
\newblock Computational geometry.
\newblock In {\em Computational geometry}, pages 1--17. Springer, 1997.

\bibitem{biniaz2014optimal}
Ahmad Biniaz, Anil Maheshwari, and Michiel Smid.
\newblock An optimal algorithm for the euclidean bottleneck full steiner tree
  problem.
\newblock {\em Computational Geometry}, 47(3):377--380, 2014.

\bibitem{BjorklundHKK07}
Andreas Bj{\"o}rklund, Thore Husfeldt, Petteri Kaski, and Mikko Koivisto.
\newblock {F}ourier meets {M}{\"o}bius: fast subset convolution.
\newblock In {\em Proceedings of the 39th Annual ACM Symposium on Theory of
  Computing (STOC)}, pages 67--74, New York, 2007. ACM.

\bibitem{BlellochDHRSS06}
Guy~E. Blelloch, Kedar Dhamdhere, Eran Halperin, R.~Ravi, Russell Schwartz, and
  Srinath Sridhar.
\newblock Fixed parameter tractability of binary near-perfect phylogenetic tree
  reconstruction.
\newblock In {\em Proceedings of the 33rd International Colloquium of Automata,
  Languages and Programming (ICALP)}, volume 4051 of {\em Lecture Notes in
  Comput. Sci.}, pages 667--678. Springer, 2006.

\bibitem{brazil2014history}
Marcus Brazil, Ronald~L Graham, Doreen~A Thomas, and Martin Zachariasen.
\newblock On the history of the euclidean steiner tree problem.
\newblock {\em Archive for history of exact sciences}, 68(3):327--354, 2014.

\bibitem{brazil2015generalised}
Marcus Brazil, Charl~J Ras, Konrad~J Swanepoel, and Doreen~A Thomas.
\newblock Generalised k-steiner tree problems in normed planes.
\newblock {\em Algorithmica}, 71(1):66--86, 2015.

\bibitem{chen2001approximations}
Donghui Chen, Ding-Zhu Du, Xiao-Dong Hu, Guo-Hui Lin, Lusheng Wang, and
  Guoliang Xue.
\newblock Approximations for steiner trees with minimum number of steiner
  points.
\newblock {\em Theoretical Computer Science}, 262(1-2):83--99, 2001.

\bibitem{chiang1989powerful}
Charles~C Chiang, Majid Sarrafzadeh, and Chak-Kuen Wong.
\newblock A powerful global router: based on steiner min-max trees.
\newblock In {\em ICCAD}, pages 2--5. Citeseer, 1989.

\bibitem{DBLP:books/sp/CyganFKLMPPS15}
Marek Cygan, Fedor~V. Fomin, Lukasz Kowalik, Daniel Lokshtanov, D{\'{a}}niel
  Marx, Marcin Pilipczuk, Michal Pilipczuk, and Saket Saurabh.
\newblock {\em Parameterized Algorithms}.
\newblock Springer, 2015.
\newblock \href {https://doi.org/10.1007/978-3-319-21275-3}
  {\path{doi:10.1007/978-3-319-21275-3}}.

\bibitem{DowneyFbook13}
Rodney~G. Downey and Michael~R. Fellows.
\newblock {\em Fundamentals of Parameterized Complexity}.
\newblock Texts in Computer Science. Springer, 2013.

\bibitem{DreyfusW71}
Stuart~E. Dreyfus and Robert~A. Wagner.
\newblock The {S}teiner problem in graphs.
\newblock {\em Networks}, 1(3):195--207, 1971.

\bibitem{du2008Steiner}
Dingzhu Du and Xiaodong Hu.
\newblock {\em Steiner tree problems in computer communication networks}.
\newblock World Scientific, 2008.

\bibitem{du2002approximations}
Z~Du et~al.
\newblock Approximations for a bottleneck steiner tree problem.
\newblock {\em Algorithmica}, 32(4):554--561, 2002.

\bibitem{FlumGrohebook}
J{\"o}rg Flum and Martin Grohe.
\newblock {\em Parameterized Complexity Theory}.
\newblock Texts in Theoretical Computer Science. An EATCS Series.
  Springer-Verlag, Berlin, 2006.

\bibitem{ganley1996optimal}
Joseph~L Ganley and Jeffrey~S Salowe.
\newblock Optimal and approximate bottleneck steiner trees.
\newblock {\em Operations Research Letters}, 19(5):217--224, 1996.

\bibitem{georgakopoulos19871}
George Georgakopoulos and Christos~H Papadimitriou.
\newblock The 1-steiner tree problem.
\newblock {\em Journal of Algorithms}, 8(1):122--130, 1987.

\bibitem{GuoNW05}
Jiong Guo, Rolf Niedermeier, and Sebastian Wernicke.
\newblock Parameterized complexity of generalized vertex cover problems.
\newblock In {\em Proceedings of the 9th International Workshop Algorithms and
  Data Structures (WADS)}, volume 3608, pages 36--48. Springer, 2005.

\bibitem{hauptmann2013compendium}
Mathias Hauptmann and Marek Karpi{\'n}ski.
\newblock {\em A compendium on Steiner tree problems}.
\newblock Inst. f{\"u}r Informatik, 2013.

\bibitem{hou2010optimal}
Yung-Tsung Hou, Chia-Mei Chen, and Bingchiang Jeng.
\newblock An optimal new-node placement to enhance the coverage of wireless
  sensor networks.
\newblock {\em Wireless Networks}, 16(4):1033--1043, 2010.

\bibitem{kahng1994optimal}
Andrew~B Kahng and Gabriel Robins.
\newblock {\em On optimal interconnections for VLSI}, volume 301.
\newblock Springer Science \& Business Media, 1994.

\bibitem{karp1975computational}
Richard~M Karp.
\newblock On the computational complexity of combinatorial problems.
\newblock {\em Networks}, 5(1):45--68, 1975.

\bibitem{kedem1986union}
Klara Kedem, Ron Livne, J{\'a}nos Pach, and Micha Sharir.
\newblock On the union of jordan regions and collision-free translational
  motion amidst polygonal obstacles.
\newblock {\em Discrete \& Computational Geometry}, 1(1):59--71, 1986.

\bibitem{li2004approximation}
Zi-Mao Li, Da-Ming Zhu, and Shao-Han Ma.
\newblock Approximation algorithm for bottleneck steiner tree problem in the
  euclidean plane.
\newblock {\em Journal of Computer Science and Technology}, 19(6):791--794,
  2004.

\bibitem{lin1999steiner}
Guo-Hui Lin and Guoliang Xue.
\newblock Steiner tree problem with minimum number of steiner points and
  bounded edge-length.
\newblock {\em Information Processing Letters}, 69(2):53--57, 1999.

\bibitem{MarxPP2017}
Daniel Marx, Marcin Pilipczuk, and Micha{\l} Pilipczuk.
\newblock {On subexponential parameterized algorithms for Steiner Tree and
  Directed Subset TSP on planar graphs}.
\newblock {\em ArXiv e-prints}, July 2017.
\newblock \href {http://arxiv.org/abs/1707.02190} {\path{arXiv:1707.02190}}.

\bibitem{megiddo1983applying}
Nimrod Megiddo.
\newblock Applying parallel computation algorithms in the design of serial
  algorithms.
\newblock {\em Journal of the ACM (JACM)}, 30(4):852--865, 1983.

\bibitem{DBLP:journals/siamcomp/Mitchell99}
Joseph S.~B. Mitchell.
\newblock Guillotine subdivisions approximate polygonal subdivisions: {A}
  simple polynomial-time approximation scheme for geometric tsp, k-mst, and
  related problems.
\newblock {\em {SIAM} J. Comput.}, 28(4):1298--1309, 1999.
\newblock \href {https://doi.org/10.1137/S0097539796309764}
  {\path{doi:10.1137/S0097539796309764}}.

\bibitem{Nederlof13}
Jesper Nederlof.
\newblock Fast polynomial-space algorithms using inclusion-exclusion.
\newblock {\em Algorithmica}, 65(4):868--884, 2013.

\bibitem{pilipczuk_et_al:LIPIcs:2013:3947}
Marcin Pilipczuk, Micha\l{} Pilipczuk, Piotr Sankowski, and Erik~Jan van
  Leeuwen.
\newblock Subexponential-time parameterized algorithm for {S}teiner tree on
  planar graphs.
\newblock In {\em Proceedings of the 30th International Symposium on
  Theoretical Aspects of Computer Science (STACS)}, volume~20 of {\em Leibniz
  International Proceedings in Informatics (LIPIcs)}, pages 353--364, Dagstuhl,
  Germany, 2013. Schloss Dagstuhl--Leibniz-Zentrum fuer Informatik.

\bibitem{PilipczukPSL14}
Marcin Pilipczuk, Michal Pilipczuk, Piotr Sankowski, and Erik~Jan van Leeuwen.
\newblock Network sparsification for {S}teiner problems on planar and
  bounded-genus graphs.
\newblock In {\em Proceedings of the 55th Annual Symposium on Foundations of
  Computer Science (FOCS)}, pages 276--285. {IEEE}, 2014.

\bibitem{sarrafzadeh1992bottleneck}
Majid Sarrafzadeh and CK~Wong.
\newblock Bottleneck steiner trees in the plane.
\newblock {\em IEEE Transactions on Computers}, 41(03):370--374, 1992.

\bibitem{wang2002approximation}
Lusheng Wang and Zimao Li.
\newblock An approximation algorithm for a bottleneck k-steiner tree problem in
  the euclidean plane.
\newblock {\em Information Processing Letters}, 81(3):151--156, 2002.

\end{thebibliography}

% \appendix
% \input{appendix}

\end{document}